%% file: main.tex
\documentclass[conference,9pt]{IEEEtran}
\IEEEoverridecommandlockouts
\usepackage{cite}
\usepackage{amsmath,amssymb,amsfonts}
\usepackage{algorithmic}
\usepackage{graphicx}
\usepackage{textcomp}
\usepackage{xcolor}
\usepackage{mathtools}
\usepackage{amsthm}
\usepackage[hyphens]{url}
\usepackage[hidelinks]{hyperref}
\usepackage[hyphenbreaks]{breakurl}
\usepackage{tikz}
\usepackage{listofitems}
\usepackage{pgfplots}
\DeclareMathOperator*{\argmin}{arg\,min}

\newtheorem{proposition}{Proposition}

\pgfplotsset{compat=newest}

\def\BibTeX{{\rm B\kern-.05em{\sc i\kern-.025em b}\kern-.08em
    T\kern-.1667em\lower.7ex\hbox{E}\kern-.125emX}}

\newcommand\copyrighttext{%
  \footnotesize
  \textcopyright 2025 IEEE. Personal use of this material is permitted.  Permission from IEEE must be obtained for all other uses, in any current or future media, including reprinting/republishing this material for advertising or promotional purposes, creating new collective works, for resale or redistribution to servers or lists, or reuse of any copyrighted component of this work in other works. DOI: \href{https://doi.org/10.1109/ICASSP49660.2025.10889620}{10.1109/ICASSP49660.2025.10889620}}
\newcommand\copyrightnotice{%
\begin{tikzpicture}[remember picture,overlay]
\node[anchor=south,yshift=5pt] at (current page.south) {\fbox{\parbox{\dimexpr\textwidth-\fboxsep-\fboxrule\relax}{\copyrighttext}}};
\end{tikzpicture}%
}

\begin{document}

\title{A Comparative Study of Invariance-Aware Loss Functions for Deep Learning-based Gridless Direction-of-Arrival Estimation\thanks{This work was supported in part by IEEE Signal Processing Society
Scholarship, and in part by NSF under Grant CCF-2225617 and Grant CCF-2124929.}
}

\author{\IEEEauthorblockN{Kuan-Lin Chen}
\IEEEauthorblockA{\textit{Department of Electrical and Computer Engineering} \\
\textit{University of California, San Diego}\\
\texttt{kuc029@ucsd.edu}}
\and
\IEEEauthorblockN{Bhaskar D. Rao}
\IEEEauthorblockA{\textit{Department of Electrical and Computer Engineering} \\
\textit{University of California, San Diego}\\
\texttt{brao@ucsd.edu}}
}

\maketitle
\copyrightnotice
\begin{abstract}
Covariance matrix reconstruction has been the most widely used guiding objective in gridless direction-of-arrival (DoA) estimation for sparse linear arrays. Many semidefinite programming (SDP)-based methods fall under this category. Although deep learning-based approaches enable the construction of more sophisticated objective functions, most methods still rely on covariance matrix reconstruction. In this paper, we propose new loss functions that are invariant to the scaling of the matrices and provide a comparative study of losses with varying degrees of invariance. The proposed loss functions are formulated based on the scale-invariant signal-to-distortion ratio between the target matrix and the Gram matrix of the prediction. Numerical results show that a scale-invariant loss outperforms its non-invariant counterpart but is inferior to the recently proposed subspace loss that is invariant to the change of basis. These results provide evidence that designing loss functions with greater degrees of invariance is advantageous in deep learning-based gridless DoA estimation.
\end{abstract}


\input{intro}
\input{assumptions}
\input{proposed_methods}
\input{numerical_results}
\input{conclusion}

\bibliographystyle{IEEEtran}
\bibliography{IEEEabrv,ref}

\end{document}

%% file: intro.tex
\section{Introduction}
DoA estimation aims to determine the directions or angles of source signals in space using a sensor array. It is a key topic in array processing and a crucial component in real-world applications such as microphone array processing \cite{chen2023dnn}, radar systems \cite{zhang2023joint}, and wireless communications \cite{gaber2014study}. To achieve greater degrees of freedom, sparse linear arrays (SLAs) have garnered significant attention, including nested arrays \cite{pal2010nested} and minimum redundancy arrays (MRAs) \cite{van2002optimum}. With an SLA that has no ``holes'' in its co-array, it is possible to localize more sources than sensors, identifying up to $\mathcal{O}(n^2)$ uncorrelated sources with only $n$ sensors \cite{pal2010nested}. Taking MRAs for example, $4$-element and $6$-element MRAs can resolve up to $6$ and $13$ sources, respectively. In the gridless case, these extra degrees of freedom are often achieved by augmenting the spatial covariance matrix from an SLA to the covariance of its corresponding virtual uniform linear array (ULA) via some fitting criteria and constraints \cite{li2015off,qiao2017maximum,wang2019grid}. Methods based on convex optimization in this category include redundancy averaging and direct augmentation \cite{pillai1985new}, sparse and parametric approach (SPA) \cite{yang2014discretization}, and StructCovMLE \cite{pote2023maximum}, just to name a few.

While optimization-based methods have long dominated the field \cite{stoica2010new,stoica2010spice,stoica2012spice,zhou2018direction}, recent advances have seen the rise of deep learning techniques as powerful alternatives \cite{liu2018direction,papageorgiou2021deep}.
These emerging approaches formulate covariance matrix reconstruction as supervised learning problems and often train deep neural network (DNN) models to fit the target covariance matrices. In \cite{wu2022gridless}, a squared loss function was proposed to fit the Toeplitz matrix formed from a convolutional neural network to the noiseless co-array covariance matrix. In \cite{barthelme2021doa}, the Frobenius norm and the \textit{affine invariant distance} were used to fit the Gram matrix of the matrix output of a DNN to the ground-truth covariance matrix. These deep learning-based methods have shown promising results over widely used optimization-based methods such as SPA. One possible explanation for why a DNN can achieve superior performance is that the mapping between the covariance of an SLA and the co-array covariance is continuous for a number of sources. Because a DNN model can approximate a continuous function on a compact subset of a Euclidean space \cite{cybenko1989approximation,hornik1989multilayer,chen2022improved}, it may be possible for a DNN to estimate the co-array covariance matrix.

Instead of relying on covariance matrix fitting, a recently proposed methodology called \textit{subspace representation learning} \cite{chen2024subspace} showed that it is possible to directly learn the signal and noise subspaces from the sample covariance of an SLA by minimizing a collection of loss functions imposed on subspaces viewed on a union of Grassmann manifolds. These loss functions are designed to possess several invariance properties to expand the solution space of the DNN model. Specifically, \cite{chen2024subspace} points out that covariance matrix reconstruction is actually a more challenging problem than learning signal and noise subspaces. This difficulty arises from the need to estimate signal powers when reconstructing a spatial covariance matrix, a requirement not present in DoA estimation, where the MUltiple SIgnal Classification (MUSIC) \cite{schmidt1986multiple} or root-MUSIC algorithms \cite{barabell1983improving,rao1989performance} rely solely on knowledge of the signal or noise subspace.

In this paper, we provide a comparative study of loss functions that exploit different degrees of invariance. In particular, we propose a new family of scale-invariant loss functions based on the scale-invariant signal-to-distortion ratio (SI-SDR) \cite{le2019sdr} and show that it expands the solution space of the widely used Frobenius norm. In addition, we analyze the invariance properties of two recently proposed loss functions that measure the length of the shortest curve between two elements in an appropriate space. Numerical results show that the proposed scale-invariant covariance reconstruction loss is superior to the Frobenius norm. However, scale-invariant losses are inferior to the ``subspace loss'' \cite{chen2024subspace} that yields the greatest degrees of invariance. These results imply that a more comprehensive invariance facilitates the learning process of a DNN model, serving as evidence on the benefits of invariance-aware losses for gridless DoA estimation.

%% file: assumptions.tex
\section{Assumptions and the signal model} \label{sec:assumptions}
Consider $k$ far-field narrowband source signals $s_1,s_2,\cdots,s_k$ with wavelength $\lambda$ impinging on an $m$-element ULA with spacing $d=\frac{\lambda}{2}$ from directions $\left\{\theta_1,\theta_2,\cdots,\theta_k\right\}\subset[0,\pi]$. Let the array manifold be $\mathbf{a}:[0,\pi]\to\mathbb{C}^m$ such that
$
    [\mathbf{a}(\theta)]_{i}\coloneqq
    e^{j2\pi\left(i-1-\frac{m-1}{2}\right)\frac{d}{\lambda}\cos\theta}
$
for $i\in\{1,2,\cdots,m\}\coloneqq [m]$. The received signal or \textit{snapshot} at $t\in[T]$ then can be modeled as
$
    \mathbf{y}(t)=\mathbf{A}(\boldsymbol{\theta})\mathbf{s}(t)+\mathbf{n}(t)
$
where $\mathbf{s}(t)=\begin{bmatrix}s_1(t)&s_2(t)&\cdots&s_k(t)\end{bmatrix}^{\mathsf{T}}$ is a random source signal vector sampled from the complex  circularly-symmetric Gaussian distribution $\mathcal{CN}\left(\mathbf{0},\text{diag}(p_1,p_2,\cdots,p_k)\right)$, $\mathbf{n}(t)$ is a random noise vector sampled from $\mathcal{CN}\left(\mathbf{0},\eta\mathbf{I}\right)$, and $\mathbf{A}(\boldsymbol{\theta})=\begin{bmatrix}\mathbf{a}(\theta_1)&\mathbf{a}(\theta_2)&\cdots&\mathbf{a}(\theta_k)\end{bmatrix}$. All source signal vectors $\mathbf{s}(1),\mathbf{s}(2),\cdots,\mathbf{s}(T)$ and noise vectors $\mathbf{n}(1),\mathbf{n}(2),\cdots,\mathbf{n}(T)$ are independent and identically distributed (i.i.d.), respectively. $\mathbf{s}(t_1)$ and $\mathbf{n}(t_2)$ are uncorrelated for all $t_1,t_2\in[T]$. Let $\mathcal{S}=\{e_1,e_2,\cdots,e_n\}\subseteq[m]$. Then an $n$-element SLA can be created by selecting a subset of sensors ordered in $\mathcal{S}$. Let
\begin{equation}
    \left[\boldsymbol{\Gamma}\right]_{ij}=
    \begin{cases}
    1, &\text{ if } e_i=j,\\
    0, &\text{ otherwise},
    \end{cases}, i\in[n], j\in[m].
\end{equation}
The snapshot received at the SLA at $t\in[T]$ is then given by $\mathbf{y}_{\mathcal{S}}(t)=\boldsymbol{\Gamma}\mathbf{y}(t)$ and the sample covariance matrix at this SLA is $\hat{\mathbf{R}}_{\mathcal{S}}=\frac{1}{T}\sum_{t=1}^T\mathbf{y}_{\mathcal{S}}(t)\mathbf{y}_{\mathcal{S}}^{\mathsf{H}}(t)$. The noiseless covariance matrix at the corresponding ULA is denoted by $\mathbf{R}$. We are interested in estimating the directions $\theta_1,\theta_2,\cdots,\theta_k$ given $\hat{\mathbf{R}}_{\mathcal{S}}$ and $k$ under $T\geq m$.

%% file: proposed_methods.tex
\section{Scale-invariant Loss Functions}
Training a deep learning-based covariance matrix reconstruction model can be formulated as minimizing the empirical risk
\begin{equation} \label{erm:first}
    \min_{W} \quad \frac{1}{L}\sum_{l=1}^Ld\left(g\circ f_W\left(\hat{\mathbf{R}}_{\mathcal{S}}^{(l)}\right),h\left(\mathbf{R}^{(l)}\right)\right)
\end{equation}
where $f_W:\mathbb{C}^{n\times n}\to\mathbb{C}^{m\times m}$ is a DNN model with parameters $W$, $\{\hat{\mathbf{R}}_{\mathcal{S}}^{(l)},\mathbf{R}^{(l)}\}_{l=1}^L$ is a dataset of sample covariance matrices at the SLA and noiseless covariance matrices at the corresponding ULA, $h$ is a function that extracts the learning target, and $g$ is a transformation that ensures some properties of a valid covariance matrix. For example, picking the function
\begin{equation} \label{eq:positive_definite_g}
    g(\mathbf{E})=\mathbf{E}\mathbf{E}^{\mathsf{H}}+\delta\mathbf{I}
\end{equation}
for some $\delta\geq 0$ enforces the predicted matrix being always positive semidefinite (or positive definite). One can also design $g$ such that the prediction is Toeplitz.
$d:\mathbb{C}^{m\times m}\times \mathbb{C}^{m\times m}\to[0,\infty)$ is a loss function of choice.
The framework of (\ref{erm:first}) is generic in the sense that it encapsulates different empirical risk minimization problems proposed in \cite{barthelme2021doa,wu2022gridless,chen2024subspace}. The most intuitive loss function is the Frobenius norm
\begin{equation}
    d_{\text{Fro}}\left(\hat{\mathbf{R}},\mathbf{R}\right)=\left\lVert \hat{\mathbf{R}}-\mathbf{R}\right\rVert_F
\end{equation}
which yields a solution space of a single point in $\mathbb{C}^{m\times m}$ as $d_{\text{Fro}}(\hat{\mathbf{R}},\mathbf{R})=0$ if and only if $\hat{\mathbf{R}}=\mathbf{R}$. However, $\alpha\mathbf{R}$ for any $\alpha\in\mathbb{R}\setminus\{0\}$ leads to identical signal and noise subspaces, while $d_{\text{Fro}}(\alpha\mathbf{R},\mathbf{R})\to\infty$ as $\alpha\to\infty$ or $\alpha\to-\infty$ for any positive definite matrix $\mathbf{R}$. To avoid such a penalization and allow a larger solution space, we propose the following \textit{scale-invariant reconstruction loss}:
\begin{equation} \label{eq:si_loss}
    d_{\text{SI}}\left(\hat{\mathbf{R}},\mathbf{R}\right)=-\log\left(\frac{\left\lVert\alpha^*\mathbf{R}\right\rVert_F}{\epsilon+\left\lVert\alpha^*\mathbf{R}-\hat{\mathbf{R}}\right\rVert_F}\right)
\end{equation}
where $\epsilon \geq 0$ is a constant and
\begin{equation}
    \alpha^*=\argmin_{\alpha\in\mathbb{R}} \left\lVert\alpha\mathbf{R}-\hat{\mathbf{R}}\right\rVert_F.
\end{equation}
The scale-invariant reconstruction loss $d_{\text{SI}}$ is invariant to scaling of the matrices in the following sense:
\begin{equation}
    d_{\text{SI}}\left(\gamma\mathbf{R},\mathbf{R}\right)\to-\infty \quad \text{as} \quad \epsilon\to 0
\end{equation}
for every $\gamma\neq 0$.
Approximately, this property allows $d_{\text{SI}}$ to expand the solution space from a point to a line in $\mathbb{C}^{m\times m}$. If $\epsilon>0$, in general we have $d_{\text{SI}}\left(\gamma\mathbf{R}_1,\mathbf{R}_1\right)\neq d_{\text{SI}}\left(\gamma\mathbf{R}_2,\mathbf{R}_2\right)$ for $\mathbf{R}_1\neq\mathbf{R}_2$. This implies that different penalties are given to two fittings even though they lead to the correct subspaces. The notion of scale-invariant reconstruction extends beyond fitting covariance matrices. Denoting $\mathbf{E}_s$ as a matrix whose columns are eigenvectors of $\mathbf{R}$, the scale-invariant reconstruction loss can be applied to the signal subspace matrix $h(\mathbf{R})=\mathbf{E}_s\mathbf{E}_s^{\mathsf{H}}$ as follows
\begin{equation} \label{eq:si_loss_sig}
-\log\left(\frac{\left\lVert\alpha^*\mathbf{E}_s\mathbf{E}_s^{\mathsf{H}}\right\rVert_F}{\epsilon+\left\lVert\alpha^*\mathbf{E}_s\mathbf{E}_s^{\mathsf{H}}-g\circ f_W\left(\hat{\mathbf{R}}_{\mathcal{S}}\right)\right\rVert_F}\right)
\end{equation}
where
\begin{equation} \label{eq:si_loss_sig_alpha}
        \alpha^*=\argmin_{\alpha\in\mathbb{R}}\left\lVert\alpha\mathbf{E}_s\mathbf{E}_s^{\mathsf{H}}-g\circ f_W\left(\hat{\mathbf{R}}_{\mathcal{S}}\right)\right\rVert_F.
    \end{equation}
The same formulation as in (\ref{eq:si_loss_sig}) and (\ref{eq:si_loss_sig_alpha}) can be applied to the noise subspace, where $h(\mathbf{R})=\mathbf{E}_n\mathbf{E}_n^{\mathsf{H}}$ and $\mathbf{E}_n$ denotes the noise subspace.

\section{Geodesic distances and their invariances}
Because all covariance targets with a proper $h$ can reside in the positive definite cone, one can use the affine invariant distance \cite{barthelme2021doa}
\begin{equation} \label{eq:affine_invariant_distance}
    d_{\text{Aff}}\left(\hat{\mathbf{R}},\mathbf{R}\right)=\left\lVert\log\left(\mathbf{R}^{-\frac{1}{2}}\hat{\mathbf{R}}\mathbf{R}^{-\frac{1}{2}}\right)\right\rVert_F
\end{equation}
which measures the length of the shortest curve between two positive definite matrices \cite{bhatia2009positive} for covariance matrix fitting. The $\log$ in (\ref{eq:affine_invariant_distance}) is the matrix logarithm. Although (\ref{eq:affine_invariant_distance}) was used by \cite{barthelme2021doa} and \cite{chen2024subspace}, its invariance property does not seem to be well recognized. Proposition \ref{prop:aff_dist_scaling} shows that the affine invariant distance is not invariant to scaling but the rate of increased distance is a logarithmic growth in terms of scaling, much slower than the linear rate of the Frobenius norm.
\begin{proposition} \label{prop:aff_dist_scaling}
    For every $m$-by-$m$ Hermitian matrix such that $\mathbf{R}\succ 0$ and for every $\alpha>0$,
    \begin{equation}
        d_{\textnormal{Aff}}\left(\alpha\mathbf{R},\mathbf{R}\right)=\sqrt{m}\lvert\log\alpha\rvert.
    \end{equation}
\end{proposition}
\begin{proof}
    $d_{\text{Aff}}\left(\alpha\mathbf{R},\mathbf{R}\right)=\left\lVert\log\left(\frac{1}{\alpha}\mathbf{I}\right)\right\rVert_F=\left\lVert-\log\left(\alpha\right)\mathbf{I}\right\rVert_F.$
\end{proof}
Despite $d_{\text{Aff}}$ being not scale-invariant, its increased distance is invariant to the underlying matrix and only depends on the scaling factor, unlike the Frobenius norm, which depends on the matrix. In other words, for $\mathbf{R}_1\neq\mathbf{R}_2$, we have $d_{\text{Aff}}\left(\alpha\mathbf{R}_1,\mathbf{R}_1\right)=d_{\text{Aff}}\left(\alpha\mathbf{R}_2,\mathbf{R}_2\right)$, ensuring the same penality for perfect fittings.

Loss functions with the greatest degrees of invariance are perhaps the ones proposed in the \textit{subspace representation learning} methodology \cite{chen2024subspace} that avoids reconstructing covariance matrices. In subspace representation learning, a DNN is instructed to generate subspace representations which can be viewed as elements lying in a union of Grassmann manifolds $\bigcup_{k=1}^{M-1}\text{Gr}(k,m)$, reflecting signal subspaces of different dimensions. With these subspace representations, learning signal or noise subspaces can then be achieved by minimizing a collection of geodesic distances. The geodesic distance $d_{\text{Gr}-k}$ on $\text{Gr}(k,m)$ measures the length of the shortest path between two subspaces $\mathcal{U}_1,\mathcal{U}_2\in\text{Gr}(k,m)$, which can be calculated by
\begin{equation}
    d_{\text{Gr}-k}\left(\mathcal{U}_1,\mathcal{U}_2\right)=\sqrt{\sum_{i=1}^k\phi_i^2\left(\mathcal{U}_1,\mathcal{U}_2\right)}
\end{equation}
where $\phi_k\left(\mathcal{U}_1,\mathcal{U}_2\right)$ is the $i$-th principal angle between $\mathcal{U}_1$ and $\mathcal{U}_2$. Let $\mathbf{U}_1$ and $\mathbf{U}_2$ be the matrices whose columns form unitary bases for $\mathcal{U}_1$ and $\mathcal{U}_2$, respectively. For every $i\in[k]$, the $i$-th principal angle can be calculated by
\begin{equation}
    \phi_i\left(\mathcal{U}_1,\mathcal{U}_2\right)=\cos^{-1}\sigma_i
\end{equation}
where $\sigma_1\geq\sigma_2\geq\cdots\geq\sigma_k$ are the singular values of $\mathbf{U}_1^{\mathsf{H}}\mathbf{U}_2$.

%% file: numerical_results.tex
\input{mse_vs_snr}
\section{Numerical Results}
\input{mse_vs_snapshot}
\input{all_methods}
\paragraph{Settings} Based on the settings in Section \ref{sec:assumptions}, we use a $4$-element MRA with $\mathcal{S}=\{1,2,5,7\}$ ($n=4$ and $m=7$) and assume equal source powers. The signal-to-noise ratio (SNR) is defined as $10\log_{10}\frac{p_1}{\eta}$. For evaluation, every direction $\theta_i$ in $\theta_1,\theta_2,\cdots,\theta_k$ is uniformly sampled at random within $\left[\frac{\pi}{6},\frac{5\pi}{6}\right]$ such that $\min_{i\neq j}\lvert\theta_i-\theta_j\rvert\geq\frac{\pi}{45}$. For each $k$, we generate $100$ random vectors of directions, and for each random vector, we sample $100$ random source signals and noises, giving $N=10,000$ random trials in total. The evaluation metric is the mean squared error (MSE) defined as
$
    \frac{1}{N}\sum_{i=1}^{N}\frac{1}{k}\min_{\mathbf{P}\in\mathcal{P}_k}\left\lVert\mathbf{P}\hat{\boldsymbol{\theta}}_{i}-\boldsymbol{\theta}_{i}\right\rVert_2^2
$
where $\mathcal{P}_k$ is the set of all $k$-by-$k$ permutation matrices. $\boldsymbol{\theta}_{i}$ and $\hat{\boldsymbol{\theta}}_{i}$ are the ground-truth and prediction vectors of directions, respectively.

\paragraph{Methods} For optimization-based methods, there are numerous choices. We consider redundancy averaging and direct augmentation (DA) \cite{pillai1985new} for its simplicity and SPA \cite{yang2014discretization} for its competitive performance. The SDPT3 solver \cite{toh1999sdpt3} in CVX \cite{cvx,gb08} is used to solve the SDP problem in SPA. Regarding deep learning-based methods, we include the covariance reconstruction approach \cite{barthelme2021doa} and the subspace representation learning method (the subspace loss or ``Subspace'') \cite{chen2024subspace}. (\ref{eq:positive_definite_g}) with $\delta=0$ is used to ensure the positive semidefiniteness of the reconstructed matrix. For the covariance reconstruction approach, both the Frobenius norm and the affine invariant distance are included, and denoted by ``Cov'' and ``Cov-Aff,'' respectively. For Cov-Aff, we use $h(\mathbf{E})=\mathbf{E}+10^{-4}\mathbf{I}$ to ensure the positive definiteness. For all methods, the root-MUSIC algorithm is used to find directions. We denote the proposed loss functions in (\ref{eq:si_loss}) and (\ref{eq:si_loss_sig}) as ``SI-Cov'' and ``SI-Sig,'' respectively. $\epsilon=0$.

\paragraph{Models} 
We use the wide residual network (WRN) with $16$ layers and an expansion factor of $8$, or WRN-16-8 \cite{zagoruyko2016wide,he2016identity} without batch normalization (see \cite{chen2021resnests} for its guarantees). The number of parameters in WRN-16-8 is about 11 million. The network takes a real-valued tensor in $\mathbb{R}^{2\times n\times n}$ and outputs a real-valued tensor in $\mathbb{R}^{2\times m\times m}$. Complex matrices are converted into their corresponding real representations by splitting into real and imaginary parts, and vice versa. Only one DNN is trained for each method and the same architecture is used for all loss functions.
\paragraph{Training} The one-cycle learning rate scheduler \cite{smith2019super} is used. The batch size is $4,096$. The learning rates for Cov, SI-Cov, SI-Sig, Cov-Aff, and Subspace are $0.01$, $0.05$, $0.2$, $0.005$, and $0.1$, respectively. These learning rates were found to be optimal in terms of empirical risk on the validation set, determined through grid search. The training and validation sets contain $2\times 10^{6}$ and $6\times 10^{5}$ examples, respectively, for each $k\in[m-1]$. All examples are created at $T=50$ and follow the assumptions in Section \ref{sec:assumptions}. The directions are randomly generated within the range $\left[\frac{\pi}{6},\frac{5\pi}{6}\right]$ such that $\min_{i\neq j}\lvert\theta_i-\theta_j\rvert\geq\frac{\pi}{60}$. Every example uses a random SNR in dB sampled from the finite set $\{-11,-9,\cdots,19,21\}$. We use PyTorch \cite{paszke2019pytorch} to implement models and training. Code is available at {\fontfamily{qcr}\selectfont
\url{https://github.com/kjason/SubspaceRepresentationLearning}}.

\paragraph{Benefits of scale invariances}
Fig. \ref{fig:mse_vs_snr} shows the MSEs evaluated on SNRs $\{-10,-8,\cdots,18,20\}$ at $T=50$ snapshots. In general, SI-Cov outperforms Cov, verifying the advantage of expanding the solution space via scale invariance. In particular, SI-Cov is significantly better than Cov for $k=3,4,5$. For $k=1$, SI-Cov and Cov have the same performance. For $k=6$, SI-Cov slightly outperforms Cov. For $k=2$, SI-Cov is worse than Cov in high SNR regions due to its unstable behavior. The cause of this behavior is not known and we leave it for future work. On the other hand, SI-Sig has mixed results. It outperforms SI-Cov for $k=6$ but underperforms for $k=3$ and $k=4$. In general, all deep learning-based methods here significantly outperform DA and SPA.
Fig. \ref{fig:mse_vs_snapshots} evaluates MSEs on different numbers of snapshots at $20$ dB SNR. It shows that deep learning-based approaches can generalize their performance to unseen numbers of snapshots. Models trained at $50$ snapshots are able to generalize well from $10$ to $100$ snapshots. In comparison, SI-Cov generally outperforms Cov due to its scale-invariant property. SI-Sig underperforms SI-Cov for $k=1$ and $k=4$ but it significantly outperforms all the other methods for $k=6$.

\paragraph{On greater degrees of invariance}
Fig. \ref{fig:all_methods} compares MSEs over a wide range of SNRs at $T=50$ snapshots. Overall, subspace representation learning significantly outperforms all the other methods that rely on covariance matrix reconstruction. Because a subspace remains the same regardless of which basis is chosen, the subspace loss exhibits the greatest degrees of invariance, substantially expanding the solution space. Cov-Aff in most cases outperforms Cov, suggesting the advantage of the invariance properties ensured by Proposition \ref{prop:aff_dist_scaling}. For SI-Cov and Cov-Aff, Fig. \ref{fig:all_methods} shows mixed results, which may be attributed to their invariance properties leading to optimization landscapes of similar quality.

\paragraph{Other MRAs}
A study of the $5$-element MRA with $\mathcal{S}=\{1,2,5,8,10\}$ was also conducted although its numerical results are not included here (available in the code repository). The findings indicate that SI-Cov and SI-Sig generally outperform Cov, with the subspace representation learning method achieving the best performance among all approaches. The main conclusions drawn from the $4$-element MRA are consistent with those of the $5$-element MRA.

%% file: mse_vs_snr.tex
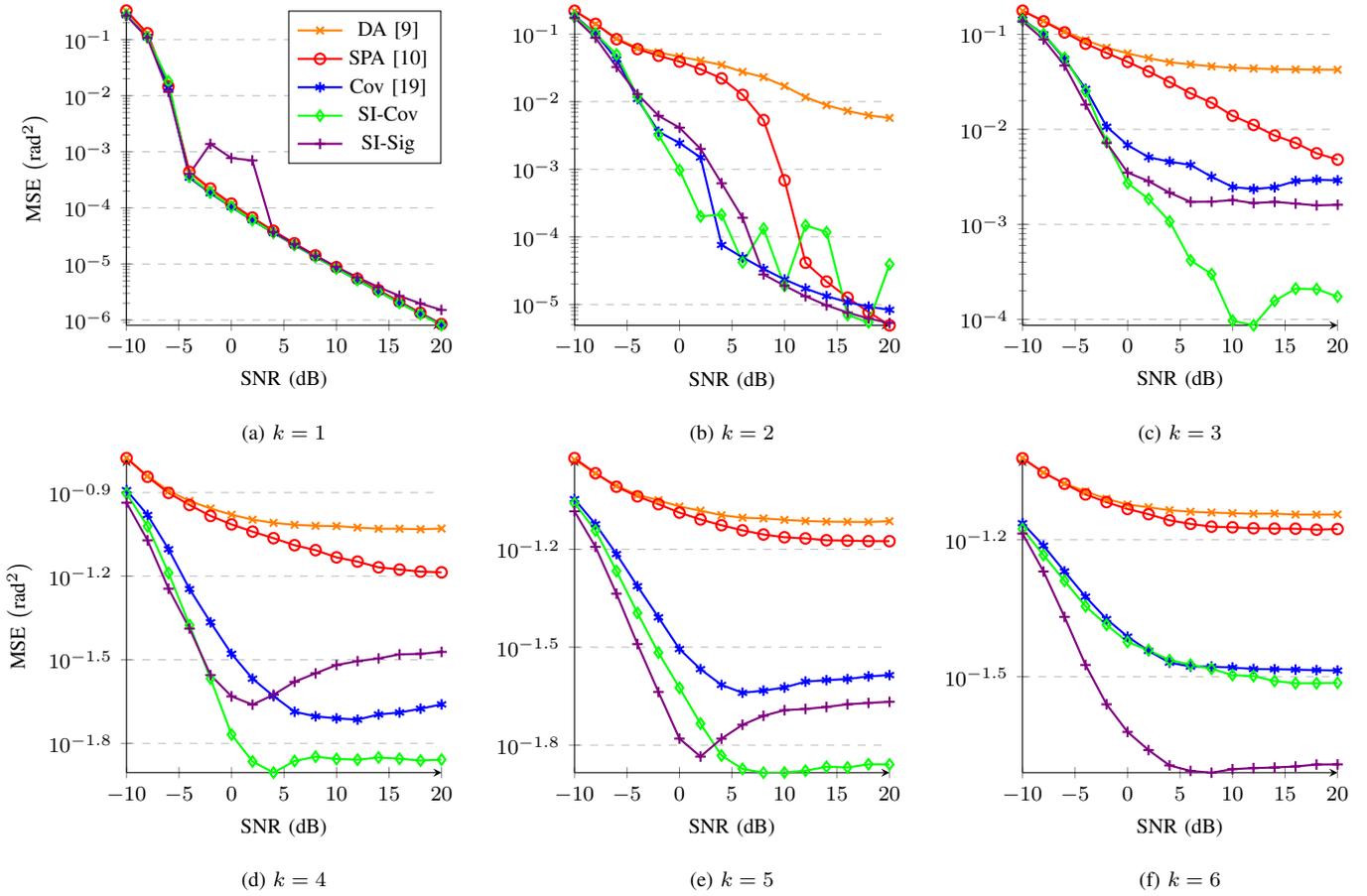
\begin{figure*}[!htb]
\centering
\begin{tikzpicture}
\begin{axis}[
    font=\footnotesize,
    width=5.8cm,
    height=5.8cm,
    axis lines = left,
    xlabel = SNR (dB),
    ylabel = MSE $\left(\text{rad}^2\right)$,
    legend style={at={(1.0,1.0)},anchor=north east},
    xtick={-10,-5,0,5,10,15,20},
    title={\footnotesize (a) $k=1$},
    title style={at={(0.5,-0.45)},anchor=south},
    xmax = 20,
    ymajorgrids=true,
    grid style=dashed,
    ymode=log
]

\addplot [
    color=orange,
    mark=x,
    thick
]
coordinates {
(-10,0.32149)
(-8,0.1259)
(-6,0.013591)
(-4,0.00042253)
(-2,0.00021344)
(0,0.00011506)
(2,6.4946e-05)
(4,3.7921e-05)
(6,2.2687e-05)
(8,1.3806e-05)
(10,8.5e-06)
(12,5.2748e-06)
(14,3.2908e-06)
(16,2.0603e-06)
(18,1.293e-06)
(20,8.1278e-07)
};

\addlegendentry{DA \cite{pillai1985new}}

\addplot [
    color=red,
    mark=o,
    thick
]
coordinates {
(-10,0.32959)
(-8,0.12969)
(-6,0.014528)
(-4,0.0004392)
(-2,0.00022123)
(0,0.00011896)
(2,6.6996e-05)
(4,3.9059e-05)
(6,2.3361e-05)
(8,1.4226e-05)
(10,8.7737e-06)
(12,5.4565e-06)
(14,3.4123e-06)
(16,2.1415e-06)
(18,1.347e-06)
(20,8.4838e-07)
};

\addlegendentry{SPA \cite{yang2014discretization}}

\addplot [
    color=blue,
    mark=asterisk,
    thick
]
coordinates {
(-10,0.27403)
(-8,0.11084)
(-6,0.013844)
(-4,0.00034608)
(-2,0.00018285)
(0,0.0001034)
(2,6.0278e-05)
(4,3.5772e-05)
(6,2.1733e-05)
(8,1.3373e-05)
(10,8.2899e-06)
(12,5.1725e-06)
(14,3.2435e-06)
(16,2.0462e-06)
(18,1.2974e-06)
(20,8.287e-07)
};
\addlegendentry{Cov \cite{barthelme2021doa}}

\addplot [
    color=green,
    mark=diamond,
    thick
]
coordinates {
(-10,0.27432)
(-8,0.11058)
(-6,0.018142)
(-4,0.00036085)
(-2,0.00018841)
(0,0.00010473)
(2,6.0455e-05)
(4,3.5987e-05)
(6,2.1809e-05)
(8,1.3388e-05)
(10,8.2935e-06)
(12,5.1577e-06)
(14,3.2246e-06)
(16,2.0245e-06)
(18,1.2771e-06)
(20,8.0926e-07)
};

\addlegendentry{SI-Cov}

\addplot [
    color=violet,
    mark=+,
    thick
]
coordinates {
(-10,0.26699)
(-8,0.10784)
(-6,0.011783)
(-4,0.00040134)
(-2,0.0013812)
(0,0.00077321)
(2,0.00070049)
(4,3.7098e-05)
(6,2.2569e-05)
(8,1.4002e-05)
(10,8.8805e-06)
(12,5.7811e-06)
(14,3.8813e-06)
(16,2.7037e-06)
(18,1.9716e-06)
(20,1.5159e-06)
};
\addlegendentry{SI-Sig}

\end{axis}

\begin{axis}[
    font=\footnotesize,
    width=5.8cm,
    height=5.8cm,
    at={(6cm,0cm)},
    axis lines = left,
    xlabel = SNR (dB),
    %ylabel = MSE $\left(\text{rad}^2\right)$,
    xmax = 20,
    xtick={-10,-5,0,5,10,15,20},
    title={\footnotesize (b) $k=2$},
    title style={at={(0.5,-0.45)},anchor=south},
    ymajorgrids=true,
    grid style=dashed,
    %ymin = 0.01,
    %ymax= 1000000000000
    ymode=log
]

\addplot [
    color=orange,
    mark=x,
    thick
]
coordinates {
(-10,0.21932)
(-8,0.141)
(-6,0.08633)
(-4,0.062491)
(-2,0.053697)
(0,0.046362)
(2,0.040376)
(4,0.035005)
(6,0.027577)
(8,0.023047)
(10,0.017014)
(12,0.011681)
(14,0.0088765)
(16,0.0072971)
(18,0.0062721)
(20,0.0057518)
};

\addplot [
    color=red,
    mark=o,
    thick
]
coordinates {
(-10,0.22133)
(-8,0.14157)
(-6,0.083972)
(-4,0.059928)
(-2,0.047836)
(0,0.039159)
(2,0.030224)
(4,0.022114)
(6,0.012627)
(8,0.0053382)
(10,0.00068585)
(12,4.1451e-05)
(14,2.1805e-05)
(16,1.2606e-05)
(18,7.7259e-06)
(20,4.899e-06)
};

\addplot [
    color=blue,
    mark=asterisk,
    thick
]
coordinates {
(-10,0.18579)
(-8,0.10209)
(-6,0.042301)
(-4,0.010954)
(-2,0.0036017)
(0,0.002439)
(2,0.0014945)
(4,7.5926e-05)
(6,4.975e-05)
(8,3.355e-05)
(10,2.3451e-05)
(12,1.7205e-05)
(14,1.3286e-05)
(16,1.0818e-05)
(18,9.2696e-06)
(20,8.3221e-06)
};

\addplot [
    color=green,
    mark=diamond,
    thick
]
coordinates {
(-10,0.18019)
(-8,0.10141)
(-6,0.049856)
(-4,0.011877)
(-2,0.0032144)
(0,0.00097577)
(2,0.0002006)
(4,0.0002135)
(6,4.1946e-05)
(8,0.00013252)
(10,1.8851e-05)
(12,0.00014798)
(14,0.00011747)
(16,6.9654e-06)
(18,5.4302e-06)
(20,3.9211e-05)
};

\addplot [
    color=violet,
    mark=+,
    thick
]
coordinates {
(-10,0.17319)
(-8,0.087977)
(-6,0.032618)
(-4,0.012934)
(-2,0.0061911)
(0,0.004117)
(2,0.002002)
(4,0.0006198)
(6,0.00019154)
(8,2.781e-05)
(10,1.8922e-05)
(12,1.3202e-05)
(14,9.7094e-06)
(16,7.5821e-06)
(18,6.1994e-06)
(20,5.3116e-06)
};
\end{axis}

\begin{axis}[
    font=\footnotesize,
    width=5.8cm,
    height=5.8cm,
    at={(12cm,0cm)},
    axis lines = left,
    xlabel = SNR (dB),
    %ylabel = MSE $\left(\text{rad}^2\right)$,
    xmax = 20,
    xtick={-10,-5,0,5,10,15,20},
    title={\footnotesize (c) $k=3$},
    title style={at={(0.5,-0.45)},anchor=south},
    ymajorgrids=true,
    grid style=dashed,
    %ymin = 1,
    %ymax= 1000000000000
    ymode=log
]

\addplot [
    color=orange,
    mark=x,
    thick
]
coordinates {
(-10,0.17622)
(-8,0.13843)
(-6,0.10786)
(-4,0.086616)
(-2,0.07221)
(0,0.06309)
(2,0.056258)
(4,0.050897)
(6,0.048287)
(8,0.046068)
(10,0.044585)
(12,0.043737)
(14,0.043058)
(16,0.042758)
(18,0.042527)
(20,0.042428)
};

\addplot [
    color=red,
    mark=o,
    thick
]
coordinates {
(-10,0.17748)
(-8,0.13696)
(-6,0.1048)
(-4,0.079731)
(-2,0.063749)
(0,0.051303)
(2,0.040709)
(4,0.031466)
(6,0.023971)
(8,0.019057)
(10,0.013905)
(12,0.011135)
(14,0.0086113)
(16,0.0071971)
(18,0.0055991)
(20,0.0048057)
};
\addplot [
    color=blue,
    mark=asterisk,
    thick
]
coordinates {
(-10,0.14311)
(-8,0.099995)
(-6,0.055025)
(-4,0.02672)
(-2,0.010665)
(0,0.0068238)
(2,0.0050852)
(4,0.0045685)
(6,0.0042226)
(8,0.0031719)
(10,0.00247)
(12,0.002353)
(14,0.0024551)
(16,0.0028579)
(18,0.0029424)
(20,0.0029049)
};

\addplot [
    color=green,
    mark=diamond,
    thick
]
coordinates {
(-10,0.14063)
(-8,0.10066)
(-6,0.056548)
(-4,0.025184)
(-2,0.0074023)
(0,0.0027177)
(2,0.0018478)
(4,0.0010737)
(6,0.00041766)
(8,0.0002996)
(10,9.6652e-05)
(12,8.662e-05)
(14,0.00015673)
(16,0.00021039)
(18,0.00020809)
(20,0.00017354)
};
\addplot [
    color=violet,
    mark=+,
    thick
]
coordinates {
(-10,0.13659)
(-8,0.088331)
(-6,0.046897)
(-4,0.018225)
(-2,0.0071983)
(0,0.0035061)
(2,0.0028447)
(4,0.0021509)
(6,0.0017243)
(8,0.001733)
(10,0.0017995)
(12,0.0016749)
(14,0.0017251)
(16,0.0016551)
(18,0.0015838)
(20,0.0016072)
};
\end{axis}

\begin{axis}[
    font=\footnotesize,
    width=5.8cm,
    height=5.8cm,
    at={(0cm,-6cm)},
    axis lines = left,
    xlabel = SNR (dB),
    ylabel = MSE $\left(\text{rad}^2\right)$,
    xmax = 20,
    xtick={-10,-5,0,5,10,15,20},
    title={\footnotesize (d) $k=4$},
    title style={at={(0.5,-0.45)},anchor=south},
    ymajorgrids=true,
    grid style=dashed,
    %ymin = 1,
    %ymax= 1000000000000
    ytick={10^(-1.8),10^(-1.5),10^(-1.2),10^(-0.9)},
    ymode=log
]

\addplot [
    color=orange,
    mark=x,
    thick
]
coordinates {
(-10,0.16738)
(-8,0.14431)
(-6,0.12755)
(-4,0.11742)
(-2,0.11032)
(0,0.10478)
(2,0.10061)
(4,0.097953)
(6,0.096413)
(8,0.095598)
(10,0.095404)
(12,0.094225)
(14,0.093308)
(16,0.093346)
(18,0.092937)
(20,0.093461)
};

\addplot [
    color=red,
    mark=o,
    thick
]
coordinates {
(-10,0.16726)
(-8,0.14328)
(-6,0.12544)
(-4,0.11355)
(-2,0.10373)
(0,0.096713)
(2,0.091096)
(4,0.086303)
(6,0.081437)
(8,0.078186)
(10,0.073667)
(12,0.071142)
(14,0.067854)
(16,0.066641)
(18,0.065481)
(20,0.065101)
};

\addplot [
    color=blue,
    mark=asterisk,
    thick
]
coordinates {
(-10,0.12834)
(-8,0.10454)
(-6,0.078785)
(-4,0.056708)
(-2,0.043089)
(0,0.033224)
(2,0.027041)
(4,0.023417)
(6,0.020578)
(8,0.019815)
(10,0.019482)
(12,0.019281)
(14,0.02012)
(16,0.02045)
(18,0.021105)
(20,0.021867)
};

\addplot [
    color=green,
    mark=diamond,
    thick
]
coordinates {
(-10,0.12494)
(-8,0.095219)
(-6,0.064789)
(-4,0.042099)
(-2,0.027083)
(0,0.017045)
(2,0.013648)
(4,0.012448)
(6,0.013704)
(8,0.014201)
(10,0.013935)
(12,0.013856)
(14,0.014104)
(16,0.013962)
(18,0.013757)
(20,0.013863)
};
\addplot [
    color=violet,
    mark=+,
    thick
]
coordinates {
(-10,0.11566)
(-8,0.084767)
(-6,0.056867)
(-4,0.040928)
(-2,0.027865)
(0,0.023367)
(2,0.021822)
(4,0.023685)
(6,0.02636)
(8,0.028269)
(10,0.030277)
(12,0.031257)
(14,0.03197)
(16,0.033064)
(18,0.033212)
(20,0.033746)
};
\end{axis}

\begin{axis}[
    font=\footnotesize,
    width=5.8cm,
    height=5.8cm,
    at={(6cm,-6cm)},
    axis lines = left,
    xlabel = SNR (dB),
    %ylabel = MSE $\left(\text{rad}^2\right)$,
    xmax = 20,
    xtick={-10,-5,0,5,10,15,20},
    title={\footnotesize (e) $k=5$},
    title style={at={(0.5,-0.45)},anchor=south},
    ymajorgrids=true,
    grid style=dashed,
    %ymin = 1,
    %ymax= 1000000000000
    ytick={10^(-1.8),10^(-1.5),10^(-1.2)},
    ymode=log
]

\addplot [
    color=orange,
    mark=x,
    thick
]
coordinates {
(-10,0.1187)
(-8,0.10775)
(-6,0.098873)
(-4,0.092633)
(-2,0.089197)
(0,0.085315)
(2,0.082889)
(4,0.080248)
(6,0.078893)
(8,0.07847)
(10,0.077674)
(12,0.07707)
(14,0.076695)
(16,0.076605)
(18,0.076523)
(20,0.076946)
};

\addplot [
    color=red,
    mark=o,
    thick
]
coordinates {
(-10,0.12008)
(-8,0.10792)
(-6,0.098096)
(-4,0.091661)
(-2,0.086733)
(0,0.081738)
(2,0.077888)
(4,0.074771)
(6,0.072125)
(8,0.070035)
(10,0.068686)
(12,0.068125)
(14,0.067253)
(16,0.067075)
(18,0.066848)
(20,0.066806)
};

\addplot [
    color=blue,
    mark=asterisk,
    thick
]
coordinates {
(-10,0.089885)
(-8,0.07534)
(-6,0.061011)
(-4,0.048657)
(-2,0.039045)
(0,0.031222)
(2,0.027094)
(4,0.024239)
(6,0.022928)
(8,0.023277)
(10,0.023804)
(12,0.024805)
(14,0.025071)
(16,0.025262)
(18,0.025751)
(20,0.025976)
};

\addplot [
    color=green,
    mark=diamond,
    thick
]
coordinates {
(-10,0.087347)
(-8,0.072041)
(-6,0.054112)
(-4,0.040288)
(-2,0.030438)
(0,0.023741)
(2,0.018447)
(4,0.014725)
(6,0.013405)
(8,0.013049)
(10,0.013065)
(12,0.013236)
(14,0.013606)
(16,0.013539)
(18,0.013849)
(20,0.013832)
};
\addplot [
    color=violet,
    mark=+,
    thick
]
coordinates {
(-10,0.082503)
(-8,0.064196)
(-6,0.04614)
(-4,0.032339)
(-2,0.023058)
(0,0.016597)
(2,0.014621)
(4,0.016588)
(6,0.018297)
(8,0.01949)
(10,0.020276)
(12,0.020456)
(14,0.020744)
(16,0.021138)
(18,0.021344)
(20,0.021513)
};
\end{axis}

\begin{axis}[
    font=\footnotesize,
    width=5.8cm,
    height=5.8cm,
    at={(12cm,-6cm)},
    axis lines = left,
    xlabel = SNR (dB),
    %ylabel = MSE $\left(\text{rad}^2\right)$,
    xmax = 20,
    xtick={-10,-5,0,5,10,15,20},
    title={\footnotesize (f) $k=6$},
    title style={at={(0.5,-0.45)},anchor=south},
    ymajorgrids=true,
    grid style=dashed,
    %ymin = 0.01,
    %ymax= 1000000000000
    ytick={10^(-1.5),10^(-1.2)},
    ymode=log
]

\addplot [
    color=orange,
    mark=x,
    thick
]
coordinates {
(-10,0.095256)
(-8,0.088543)
(-6,0.083781)
(-4,0.080507)
(-2,0.077581)
(0,0.075363)
(2,0.074358)
(4,0.073245)
(6,0.07268)
(8,0.072396)
(10,0.07213)
(12,0.071942)
(14,0.07202)
(16,0.071655)
(18,0.071643)
(20,0.071672)
};

\addplot [
    color=red,
    mark=o,
    thick
]
coordinates {
(-10,0.095178)
(-8,0.088579)
(-6,0.083672)
(-4,0.079309)
(-2,0.076308)
(0,0.073705)
(2,0.071789)
(4,0.069603)
(6,0.068354)
(8,0.067403)
(10,0.067193)
(12,0.066859)
(14,0.066749)
(16,0.066723)
(18,0.066382)
(20,0.066607)
};

\addplot [
    color=blue,
    mark=asterisk,
    thick
]
coordinates {
(-10,0.068616)
(-8,0.061247)
(-6,0.053786)
(-4,0.047358)
(-2,0.042304)
(0,0.038724)
(2,0.03604)
(4,0.033969)
(6,0.033295)
(8,0.03328)
(10,0.033125)
(12,0.032887)
(14,0.032847)
(16,0.032781)
(18,0.0327)
(20,0.032624)
};

\addplot [
    color=green,
    mark=diamond,
    thick
]
coordinates {
(-10,0.066659)
(-8,0.058419)
(-6,0.051292)
(-4,0.045123)
(-2,0.041089)
(0,0.037693)
(2,0.03616)
(4,0.034388)
(6,0.033615)
(8,0.032912)
(10,0.031884)
(12,0.031729)
(14,0.030948)
(16,0.030597)
(18,0.030606)
(20,0.030667)
};

\addplot [
    color=violet,
    mark=+,
    thick
]
coordinates {
(-10,0.0651)
(-8,0.053766)
(-6,0.04277)
(-4,0.033569)
(-2,0.027519)
(0,0.023944)
(2,0.021848)
(4,0.020224)
(6,0.019672)
(8,0.019495)
(10,0.019844)
(12,0.019957)
(14,0.020009)
(16,0.020115)
(18,0.020313)
(20,0.020326)
};
\end{axis}
\end{tikzpicture}
\caption{MSE vs. SNR. The proposed scale-invariant covariance matrix reconstruction approach (SI-Cov) outperforms DA, SPA, and the deep learning-based covariance matrix reconstruction approach (Cov) using the Frobenius norm when $k>2$.
}
\label{fig:mse_vs_snr}
\end{figure*}

%% file: mse_vs_snapshot.tex
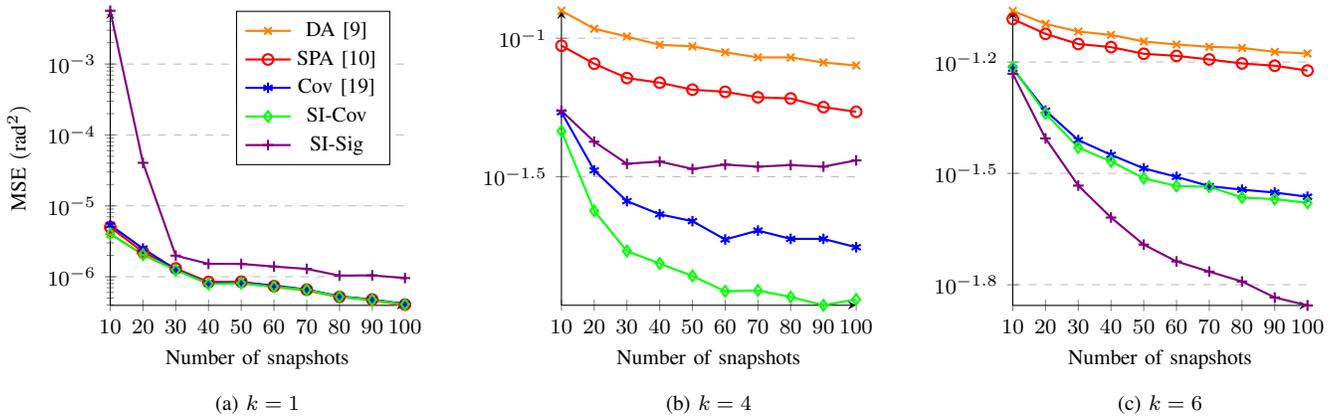
\begin{figure*}[!htb]
\centering
\begin{tikzpicture}
\begin{axis}[
    font=\footnotesize,
    width=5.5cm,
    height=5.5cm,
    axis lines = left,
    xlabel = Number of snapshots,
    ylabel = MSE $(\text{rad}^2)$,
    legend style={at={(0.95,1.0)},anchor=north east},
    xtick={10,20,30,40,50,60,70,80,90,100},
    title={\footnotesize (a) $k=1$},
    title style={at={(0.5,-0.45)},anchor=south},
    xmax = 100,
    ymajorgrids=true,
    grid style=dashed,
    ymode=log
]

\addplot [
    color=orange,
    mark=x,
    thick
]
coordinates {
(10,4.0147e-06)
(20,2.0507e-06)
(30,1.2627e-06)
(40,8.0249e-07)
(50,8.1278e-07)
(60,7.1713e-07)
(70,6.3481e-07)
(80,5.1891e-07)
(90,4.6418e-07)
(100,3.945e-07)
};

\addlegendentry{DA \cite{pillai1985new}}

\addplot [
    color=red,
    mark=o,
    thick
]
coordinates {
(10,4.9882e-06)
(20,2.2346e-06)
(30,1.308e-06)
(40,8.4195e-07)
(50,8.4831e-07)
(60,7.3582e-07)
(70,6.5667e-07)
(80,5.194e-07)
(90,4.7727e-07)
(100,4.0441e-07)
};

\addlegendentry{SPA \cite{yang2014discretization}}

\addplot [
    color=blue,
    mark=asterisk,
    thick
]
coordinates {
(10,5.347e-06)
(20,2.4982e-06)
(30,1.2649e-06)
(40,8.0955e-07)
(50,8.287e-07)
(60,7.5088e-07)
(70,6.5928e-07)
(80,5.2943e-07)
(90,4.7618e-07)
(100,4.1216e-07)
};

\addlegendentry{Cov \cite{barthelme2021doa}}

\addplot [
    color=green,
    mark=diamond,
    thick
]
coordinates {
(10,4.0176e-06)
(20,2.0454e-06)
(30,1.2407e-06)
(40,7.9246e-07)
(50,8.0926e-07)
(60,7.2215e-07)
(70,6.5258e-07)
(80,5.2122e-07)
(90,4.7012e-07)
(100,4.0668e-07)
};

\addlegendentry{SI-Cov}

\addplot [
    color=violet,
    mark=+,
    thick
]
coordinates {
(10,0.0056408)
(20,4.0404e-05)
(30,1.9838e-06)
(40,1.5207e-06)
(50,1.5159e-06)
(60,1.3914e-06)
(70,1.2892e-06)
(80,1.0358e-06)
(90,1.0457e-06)
(100,9.5656e-07)
};

\addlegendentry{SI-Sig}

\end{axis}

\begin{axis}[
    font=\footnotesize,
    width=5.5cm,
    height=5.5cm,
    at={(6cm,0cm)},
    axis lines = left,
    xlabel = Number of snapshots,
    xmax = 100,
    xtick={10,20,30,40,50,60,70,80,90,100},
    title={\footnotesize (b) $k=4$},
    title style={at={(0.5,-0.45)},anchor=south},
    ymajorgrids=true,
    grid style=dashed,
    ymode=log
]

\addplot [
    color=orange,
    mark=x,
    thick
]
coordinates {
(10,0.12562)
(20,0.10816)
(30,0.10138)
(40,0.094625)
(50,0.093461)
(60,0.088902)
(70,0.085252)
(80,0.08522)
(90,0.081699)
(100,0.079755)
};

\addplot [
    color=red,
    mark=o,
    thick
]
coordinates {
(10,0.093941)
(20,0.080919)
(30,0.071813)
(40,0.069102)
(50,0.065172)
(60,0.064046)
(70,0.061189)
(80,0.060603)
(90,0.056396)
(100,0.054247)
};

\addplot [
    color=blue,
    mark=asterisk,
    thick
]
coordinates {
(10,0.054032)
(20,0.033346)
(30,0.025821)
(40,0.023148)
(50,0.021867)
(60,0.018781)
(70,0.020224)
(80,0.018874)
(90,0.018875)
(100,0.017609)
};

\addplot [
    color=green,
    mark=diamond,
    thick
]
coordinates {
(10,0.046196)
(20,0.023829)
(30,0.017057)
(40,0.015377)
(50,0.013863)
(60,0.012209)
(70,0.012295)
(80,0.011644)
(90,0.010871)
(100,0.011416)
};
\addplot [
    color=violet,
    mark=+,
    thick
]
coordinates {
(10,0.0547)
(20,0.042323)
(30,0.03523)
(40,0.035876)
(50,0.033746)
(60,0.03499)
(70,0.034394)
(80,0.034839)
(90,0.034404)
(100,0.036233)
};
\end{axis}

\begin{axis}[
    font=\footnotesize,
    width=5.5cm,
    height=5.5cm,
    at={(12cm,0cm)},
    axis lines = left,
    xlabel = Number of snapshots,
    xmax = 100,
    xtick={10,20,30,40,50,60,70,80,90,100},
    title={\footnotesize (c) $k=6$},
    title style={at={(0.5,-0.45)},anchor=south},
    ymajorgrids=true,
    grid style=dashed,
    ytick={10^(-1.8),10^(-1.5),10^(-1.2)},
    ymode=log
]

\addplot [
    color=orange,
    mark=x,
    thick
]
coordinates {
(10,0.086636)
(20,0.079963)
(30,0.076232)
(40,0.074676)
(50,0.071672)
(60,0.070356)
(70,0.069332)
(80,0.068886)
(90,0.067196)
(100,0.066527)
};

\addplot [
    color=red,
    mark=o,
    thick
]
coordinates {
(10,0.082388)
(20,0.07515)
(30,0.070573)
(40,0.069208)
(50,0.066406)
(60,0.065555)
(70,0.064117)
(80,0.062541)
(90,0.061672)
(100,0.059843)
};

\addplot [
    color=blue,
    mark=asterisk,
    thick
]
coordinates {
(10,0.060576)
(20,0.046683)
(30,0.038886)
(40,0.035482)
(50,0.032624)
(60,0.030952)
(70,0.02922)
(80,0.028581)
(90,0.028078)
(100,0.027354)
};

\addplot [
    color=green,
    mark=diamond,
    thick
]
coordinates {
(10,0.061057)
(20,0.04585)
(30,0.037107)
(40,0.034074)
(50,0.030667)
(60,0.029239)
(70,0.029109)
(80,0.027223)
(90,0.026961)
(100,0.026343)
};
\addplot [
    color=violet,
    mark=+,
    thick
]
coordinates {
(10,0.058654)
(20,0.039293)
(30,0.029317)
(40,0.024045)
(50,0.020326)
(60,0.018313)
(70,0.017191)
(80,0.016169)
(90,0.014643)
(100,0.013939)
};
\end{axis}

\end{tikzpicture}
\caption{MSE vs. number of snapshots. Despite these models are trained with their corresponding loss functions at $T=50$ snapshots, they are able to generalize to unseen numbers of snapshots. The proposed SI-Cov outperforms Cov, showing the advantage of using the scale-invariant strategy.}
\label{fig:mse_vs_snapshots}
\end{figure*}

%% file: all_methods.tex
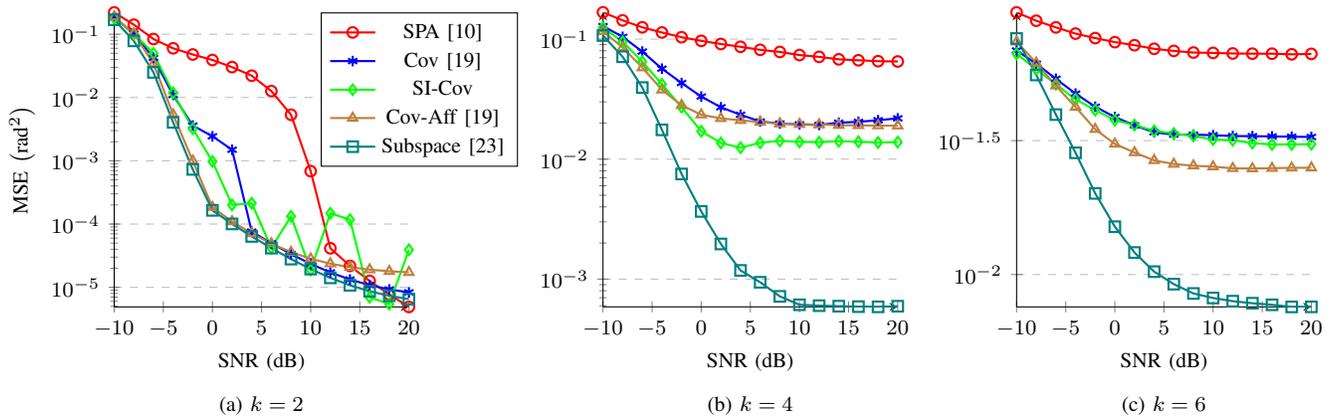
\begin{figure*}[!htb]
\centering
\begin{tikzpicture}
\begin{axis}[
    font=\footnotesize,
    width=5.5cm,
    height=5.5cm,
    at={(0cm,0cm)},
    axis lines = left,
    xlabel = SNR (dB),
    legend style={at={(1.37,1.0)},anchor=north east},
    ylabel = MSE $\left(\text{rad}^2\right)$,
    xmax = 20,
    xtick={-10,-5,0,5,10,15,20},
    title={\footnotesize (a) $k=2$},
    title style={at={(0.5,-0.45)},anchor=south},
    ymajorgrids=true,
    grid style=dashed,
    ymode=log
]

\addplot [
    color=red,
    mark=o,
    thick
]
coordinates {
(-10,0.22133)
(-8,0.14157)
(-6,0.083972)
(-4,0.059928)
(-2,0.047836)
(0,0.039159)
(2,0.030224)
(4,0.022114)
(6,0.012627)
(8,0.0053382)
(10,0.00068585)
(12,4.1451e-05)
(14,2.1805e-05)
(16,1.2606e-05)
(18,7.7259e-06)
(20,4.899e-06)
};

\addlegendentry{SPA \cite{yang2014discretization}}

\addplot [
    color=blue,
    mark=asterisk,
    thick
]
coordinates {
(-10,0.18579)
(-8,0.10209)
(-6,0.042301)
(-4,0.010954)
(-2,0.0036017)
(0,0.002439)
(2,0.0014945)
(4,7.5926e-05)
(6,4.975e-05)
(8,3.355e-05)
(10,2.3451e-05)
(12,1.7205e-05)
(14,1.3286e-05)
(16,1.0818e-05)
(18,9.2696e-06)
(20,8.3221e-06)
};

\addlegendentry{Cov \cite{barthelme2021doa}}

\addplot [
    color=green,
    mark=diamond,
    thick
]
coordinates {
(-10,0.18019)
(-8,0.10141)
(-6,0.049856)
(-4,0.011877)
(-2,0.0032144)
(0,0.00097577)
(2,0.0002006)
(4,0.0002135)
(6,4.1946e-05)
(8,0.00013252)
(10,1.8851e-05)
(12,0.00014798)
(14,0.00011747)
(16,6.9654e-06)
(18,5.4302e-06)
(20,3.9211e-05)
};

\addlegendentry{SI-Cov}

\addplot [
    color=brown,
    mark=triangle,
    thick
]
coordinates {
(-10,0.18944)
(-8,0.10209)
(-6,0.034245)
(-4,0.0053805)
(-2,0.00099794)
(0,0.00018454)
(2,0.00010949)
(4,7.1229e-05)
(6,4.8994e-05)
(8,3.598e-05)
(10,2.8298e-05)
(12,2.3659e-05)
(14,2.0808e-05)
(16,1.904e-05)
(18,1.7938e-05)
(20,1.7251e-05)
};

\addlegendentry{Cov-Aff \cite{barthelme2021doa}}

\addplot [
    color=teal,
    mark=square,
    thick
]
coordinates {
(-10,0.16951)
(-8,0.078538)
(-6,0.024939)
(-4,0.0040453)
(-2,0.00073068)
(0,0.00016467)
(2,0.00010045)
(4,6.3738e-05)
(6,4.1599e-05)
(8,2.794e-05)
(10,1.9435e-05)
(12,1.4113e-05)
(14,1.0772e-05)
(16,8.6691e-06)
(18,7.342e-06)
(20,6.5047e-06)
};

\addlegendentry{Subspace \cite{chen2024subspace}}

\end{axis}

\begin{axis}[
    font=\footnotesize,
    width=5.5cm,
    height=5.5cm,
    at={(6.5cm,0cm)},
    axis lines = left,
    xlabel = SNR (dB),
    xmax = 20,
    xtick={-10,-5,0,5,10,15,20},
    title={\footnotesize (b) $k=4$},
    title style={at={(0.5,-0.45)},anchor=south},
    ymajorgrids=true,
    grid style=dashed,
    ymode=log
]

\addplot [
    color=red,
    mark=o,
    thick
]
coordinates {
(-10,0.16726)
(-8,0.14328)
(-6,0.12544)
(-4,0.11355)
(-2,0.10373)
(0,0.096713)
(2,0.091096)
(4,0.086303)
(6,0.081437)
(8,0.078186)
(10,0.073667)
(12,0.071142)
(14,0.067854)
(16,0.066641)
(18,0.065481)
(20,0.065101)
};

\addplot [
    color=blue,
    mark=asterisk,
    thick
]
coordinates {
(-10,0.12834)
(-8,0.10454)
(-6,0.078785)
(-4,0.056708)
(-2,0.043089)
(0,0.033224)
(2,0.027041)
(4,0.023417)
(6,0.020578)
(8,0.019815)
(10,0.019482)
(12,0.019281)
(14,0.02012)
(16,0.02045)
(18,0.021105)
(20,0.021867)
};

\addplot [
    color=green,
    mark=diamond,
    thick
]
coordinates {
(-10,0.12494)
(-8,0.095219)
(-6,0.064789)
(-4,0.042099)
(-2,0.027083)
(0,0.017045)
(2,0.013648)
(4,0.012448)
(6,0.013704)
(8,0.014201)
(10,0.013935)
(12,0.013856)
(14,0.014104)
(16,0.013962)
(18,0.013757)
(20,0.013863)
};

\addplot [
    color=brown,
    mark=triangle,
    thick
]
coordinates {
(-10,0.11614)
(-8,0.085133)
(-6,0.057915)
(-4,0.037792)
(-2,0.028084)
(0,0.023459)
(2,0.021861)
(4,0.021063)
(6,0.020361)
(8,0.019739)
(10,0.019551)
(12,0.019427)
(14,0.019168)
(16,0.019112)
(18,0.019038)
(20,0.018974)
};

\addplot [
    color=teal,
    mark=square,
    thick
]
coordinates {
(-10,0.10711)
(-8,0.071279)
(-6,0.039495)
(-4,0.017544)
(-2,0.0075213)
(0,0.0036638)
(2,0.0019642)
(4,0.0011819)
(6,0.00094144)
(8,0.00071886)
(10,0.00061165)
(12,0.00060119)
(14,0.00059526)
(16,0.00059015)
(18,0.00058779)
(20,0.00059483)
};
\iffalse
\addplot [
    color=teal,
    mark=diamond,
    thick
]
coordinates {
(-10,0.11566)
(-8,0.084767)
(-6,0.056867)
(-4,0.040928)
(-2,0.027865)
(0,0.023367)
(2,0.021822)
(4,0.023685)
(6,0.02636)
(8,0.028269)
(10,0.030277)
(12,0.031257)
(14,0.03197)
(16,0.033064)
(18,0.033212)
(20,0.033746)
};
\fi
\end{axis}

\begin{axis}[
    font=\footnotesize,
    width=5.5cm,
    height=5.5cm,
    at={(12cm,0cm)},
    axis lines = left,
    xlabel = SNR (dB),
    xmax = 20,
    xtick={-10,-5,0,5,10,15,20},
    title={\footnotesize (c) $k=6$},
    title style={at={(0.5,-0.45)},anchor=south},
    ymajorgrids=true,
    grid style=dashed,
    ytick={10^(-2.0),10^(-1.5),10^(-1.0)},
    ymode=log
]

\addplot [
    color=red,
    mark=o,
    thick
]
coordinates {
(-10,0.095178)
(-8,0.088579)
(-6,0.083672)
(-4,0.079309)
(-2,0.076308)
(0,0.073705)
(2,0.071789)
(4,0.069603)
(6,0.068354)
(8,0.067403)
(10,0.067193)
(12,0.066859)
(14,0.066749)
(16,0.066723)
(18,0.066382)
(20,0.066607)
};

\addplot [
    color=blue,
    mark=asterisk,
    thick
]
coordinates {
(-10,0.068616)
(-8,0.061247)
(-6,0.053786)
(-4,0.047358)
(-2,0.042304)
(0,0.038724)
(2,0.03604)
(4,0.033969)
(6,0.033295)
(8,0.03328)
(10,0.033125)
(12,0.032887)
(14,0.032847)
(16,0.032781)
(18,0.0327)
(20,0.032624)
};

\addplot [
    color=green,
    mark=diamond,
    thick
]
coordinates {
(-10,0.066659)
(-8,0.058419)
(-6,0.051292)
(-4,0.045123)
(-2,0.041089)
(0,0.037693)
(2,0.03616)
(4,0.034388)
(6,0.033615)
(8,0.032912)
(10,0.031884)
(12,0.031729)
(14,0.030948)
(16,0.030597)
(18,0.030606)
(20,0.030667)
};

\addplot [
    color=brown,
    mark=triangle,
    thick
]
coordinates {
(-10,0.073982)
(-8,0.061633)
(-6,0.050681)
(-4,0.042184)
(-2,0.034881)
(0,0.030737)
(2,0.028546)
(4,0.026677)
(6,0.02585)
(8,0.025481)
(10,0.025281)
(12,0.024921)
(14,0.024872)
(16,0.024953)
(18,0.025024)
(20,0.025039)
};

\addplot [
    color=teal,
    mark=square,
    thick
]
coordinates {
(-10,0.07606)
(-8,0.055581)
(-6,0.039546)
(-4,0.028476)
(-2,0.020087)
(0,0.01509)
(2,0.012085)
(4,0.010253)
(6,0.0092099)
(8,0.0085001)
(10,0.008177)
(12,0.0079564)
(14,0.0078182)
(16,0.0077133)
(18,0.0075745)
(20,0.0075788)
};

\end{axis}
\end{tikzpicture}
\caption{The subspace loss outperforms all the other loss functions designed for covariance matrix reconstruction. In general, increasing the degrees of invariance leads to a better optimization landscape and yields better performance. SI-Cov and Cov-Aff have mixed results.}
\label{fig:all_methods}
\end{figure*}

%% file: conclusion.tex
\section{Conclusion}
A new family of scale-invariant loss functions is proposed for gridless DoA estimation using SLAs. With scale-invariant losses, we study several loss functions and analyze how invariance properties of a loss can play an important role in shaping the optimization landscape of a DNN model. Numerical results show that the subspace loss is better than the scale-invariant losses and the affine invariant distance, and the scale-invariant losses outperform the Frobenius norm that does not have an invariance property. These observations provide evidence that greater invariance enhances a DNN's solution space, improving performance in gridless DoA estimation.

%% file: main.bbl
\begin{thebibliography}{10}
\providecommand{\url}[1]{#1}
\csname url@samestyle\endcsname
\providecommand{\newblock}{\relax}
\providecommand{\bibinfo}[2]{#2}
\providecommand{\BIBentrySTDinterwordspacing}{\spaceskip=0pt\relax}
\providecommand{\BIBentryALTinterwordstretchfactor}{4}
\providecommand{\BIBentryALTinterwordspacing}{\spaceskip=\fontdimen2\font plus
\BIBentryALTinterwordstretchfactor\fontdimen3\font minus \fontdimen4\font\relax}
\providecommand{\BIBforeignlanguage}[2]{{%
\expandafter\ifx\csname l@#1\endcsname\relax
\typeout{** WARNING: IEEEtran.bst: No hyphenation pattern has been}%
\typeout{** loaded for the language `#1'. Using the pattern for}%
\typeout{** the default language instead.}%
\else
\language=\csname l@#1\endcsname
\fi
#2}}
\providecommand{\BIBdecl}{\relax}
\BIBdecl

\bibitem{chen2023dnn}
K.-L. Chen, C.-H. Lee, B.~D. Rao, and H.~Garudadri, ``A {DNN} based normalized time-frequency weighted criterion for robust wideband {DoA} estimation,'' in \emph{International Conference on Acoustics, Speech and Signal Processing (ICASSP)}.\hskip 1em plus 0.5em minus 0.4em\relax IEEE, 2023, pp. 1--5.

\bibitem{zhang2023joint}
X.~Zhang, Z.~Zheng, W.-Q. Wang, and H.~C. So, ``Joint {DOD} and {DOA} estimation of coherent targets for coprime {MIMO} radar,'' \emph{IEEE Transactions on Signal Processing}, vol.~71, pp. 1408--1420, 2023.

\bibitem{gaber2014study}
A.~Gaber and A.~Omar, ``A study of wireless indoor positioning based on joint {TDOA} and {DOA} estimation using {2-D} matrix pencil algorithms and {IEEE} 802.11 ac,'' \emph{IEEE Transactions on Wireless Communications}, vol.~14, no.~5, pp. 2440--2454, 2014.

\bibitem{pal2010nested}
P.~Pal and P.~P. Vaidyanathan, ``Nested arrays: A novel approach to array processing with enhanced degrees of freedom,'' \emph{IEEE Transactions on Signal Processing}, vol.~58, no.~8, pp. 4167--4181, 2010.

\bibitem{van2002optimum}
H.~L. Van~Trees, \emph{Optimum array processing: Part IV of detection, estimation, and modulation theory}.\hskip 1em plus 0.5em minus 0.4em\relax John Wiley \& Sons, 2002.

\bibitem{li2015off}
Y.~Li and Y.~Chi, ``Off-the-grid line spectrum denoising and estimation with multiple measurement vectors,'' \emph{IEEE Transactions on Signal Processing}, vol.~64, no.~5, pp. 1257--1269, 2015.

\bibitem{qiao2017maximum}
H.~Qiao and P.~Pal, ``On maximum-likelihood methods for localizing more sources than sensors,'' \emph{IEEE Signal Processing Letters}, vol.~24, no.~5, pp. 703--706, 2017.

\bibitem{wang2019grid}
M.~Wang, Z.~Zhang, and A.~Nehorai, ``Grid-less {DOA} estimation using sparse linear arrays based on wasserstein distance,'' \emph{IEEE Signal Processing Letters}, vol.~26, no.~6, pp. 838--842, 2019.

\bibitem{pillai1985new}
S.~U. Pillai, Y.~Bar-Ness, and F.~Haber, ``A new approach to array geometry for improved spatial spectrum estimation,'' \emph{Proceedings of the IEEE}, vol.~73, no.~10, pp. 1522--1524, 1985.

\bibitem{yang2014discretization}
Z.~Yang, L.~Xie, and C.~Zhang, ``A discretization-free sparse and parametric approach for linear array signal processing,'' \emph{IEEE Transactions on Signal Processing}, vol.~62, no.~19, pp. 4959--4973, 2014.

\bibitem{pote2023maximum}
R.~R. Pote and B.~D. Rao, ``Maximum likelihood-based gridless {DoA} estimation using structured covariance matrix recovery and {SBL} with grid refinement,'' \emph{IEEE Transactions on Signal Processing}, vol.~71, pp. 802--815, 2023.

\bibitem{stoica2010new}
P.~Stoica, P.~Babu, and J.~Li, ``New method of sparse parameter estimation in separable models and its use for spectral analysis of irregularly sampled data,'' \emph{IEEE Transactions on Signal Processing}, vol.~59, no.~1, pp. 35--47, 2010.

\bibitem{stoica2010spice}
------, ``{SPICE}: A sparse covariance-based estimation method for array processing,'' \emph{IEEE Transactions on Signal Processing}, vol.~59, no.~2, pp. 629--638, 2010.

\bibitem{stoica2012spice}
P.~Stoica and P.~Babu, ``{SPICE} and {LIKES}: Two hyperparameter-free methods for sparse-parameter estimation,'' \emph{Signal Processing}, vol.~92, no.~7, pp. 1580--1590, 2012.

\bibitem{zhou2018direction}
C.~Zhou, Y.~Gu, X.~Fan, Z.~Shi, G.~Mao, and Y.~D. Zhang, ``Direction-of-arrival estimation for coprime array via virtual array interpolation,'' \emph{IEEE Transactions on Signal Processing}, vol.~66, no.~22, pp. 5956--5971, 2018.

\bibitem{liu2018direction}
Z.-M. Liu, C.~Zhang, and S.~Y. Philip, ``Direction-of-arrival estimation based on deep neural networks with robustness to array imperfections,'' \emph{IEEE Transactions on Antennas and Propagation}, vol.~66, no.~12, pp. 7315--7327, 2018.

\bibitem{papageorgiou2021deep}
G.~K. Papageorgiou, M.~Sellathurai, and Y.~C. Eldar, ``Deep networks for direction-of-arrival estimation in low {SNR},'' \emph{IEEE Transactions on Signal Processing}, vol.~69, pp. 3714--3729, 2021.

\bibitem{wu2022gridless}
X.~Wu, X.~Yang, X.~Jia, and F.~Tian, ``A gridless {DOA} estimation method based on convolutional neural network with {Toeplitz} prior,'' \emph{IEEE Signal Processing Letters}, vol.~29, pp. 1247--1251, 2022.

\bibitem{barthelme2021doa}
A.~Barthelme and W.~Utschick, ``{DoA} estimation using neural network-based covariance matrix reconstruction,'' \emph{IEEE Signal Processing Letters}, vol.~28, pp. 783--787, 2021.

\bibitem{cybenko1989approximation}
G.~Cybenko, ``Approximation by superpositions of a sigmoidal function,'' \emph{Mathematics of control, signals and systems}, vol.~2, no.~4, pp. 303--314, 1989.

\bibitem{hornik1989multilayer}
K.~Hornik, M.~Stinchcombe, and H.~White, ``Multilayer feedforward networks are universal approximators,'' \emph{Neural networks}, vol.~2, no.~5, pp. 359--366, 1989.

\bibitem{chen2022improved}
K.-L. Chen, H.~Garudadri, and B.~D. Rao, ``Improved bounds on neural complexity for representing piecewise linear functions,'' in \emph{Advances in Neural Information Processing Systems (NeurIPS)}, vol.~35, 2022.

\bibitem{chen2024subspace}
K.-L. Chen and B.~D. Rao, ``Subspace representation learning for sparse linear arrays to localize more sources than sensors: A deep learning methodology,'' \emph{arXiv preprint arXiv:2408.16605}, 2024.

\bibitem{schmidt1986multiple}
R.~Schmidt, ``Multiple emitter location and signal parameter estimation,'' \emph{IEEE Transactions on Antennas and Propagation}, vol.~34, no.~3, pp. 276--280, 1986.

\bibitem{barabell1983improving}
A.~Barabell, ``Improving the resolution performance of eigenstructure-based direction-finding algorithms,'' in \emph{International Conference on Acoustics, Speech, and Signal Processing (ICASSP)}.\hskip 1em plus 0.5em minus 0.4em\relax IEEE, 1983, pp. 336--339.

\bibitem{rao1989performance}
B.~D. Rao and K.~S. Hari, ``Performance analysis of root-music,'' \emph{IEEE Transactions on Acoustics, Speech, and Signal Processing}, vol.~37, no.~12, pp. 1939--1949, 1989.

\bibitem{le2019sdr}
J.~Le~Roux, S.~Wisdom, H.~Erdogan, and J.~R. Hershey, ``{SDR}--half-baked or well done?'' in \emph{International Conference on Acoustics, Speech and Signal Processing (ICASSP)}.\hskip 1em plus 0.5em minus 0.4em\relax IEEE, 2019, pp. 626--630.

\bibitem{bhatia2009positive}
R.~Bhatia, ``Positive definite matrices,'' in \emph{Positive Definite Matrices}.\hskip 1em plus 0.5em minus 0.4em\relax Princeton University Press, 2009.

\bibitem{toh1999sdpt3}
K.-C. Toh, M.~J. Todd, and R.~H. T{\"u}t{\"u}nc{\"u}, ``{SDPT3}—{A} {MATLAB} software package for semidefinite programming, version 1.3,'' \emph{Optimization Methods and Software}, vol.~11, no. 1-4, pp. 545--581, 1999.

\bibitem{cvx}
M.~Grant and S.~Boyd, ``{CVX}: Matlab software for disciplined convex programming, version 2.1,'' \url{https://cvxr.com/cvx}, Mar. 2014.

\bibitem{gb08}
------, ``Graph implementations for nonsmooth convex programs,'' in \emph{Recent Advances in Learning and Control}, ser. Lecture Notes in Control and Information Sciences, V.~Blondel, S.~Boyd, and H.~Kimura, Eds.\hskip 1em plus 0.5em minus 0.4em\relax Springer-Verlag Limited, 2008, pp. 95--110.

\bibitem{zagoruyko2016wide}
S.~Zagoruyko and N.~Komodakis, ``Wide residual networks,'' in \emph{British Machine Vision Conference}.\hskip 1em plus 0.5em minus 0.4em\relax BMVA Press, 2016, pp. 87.1--87.12.

\bibitem{he2016identity}
K.~He, X.~Zhang, S.~Ren, and J.~Sun, ``Identity mappings in deep residual networks,'' in \emph{European Conference on Computer Vision}.\hskip 1em plus 0.5em minus 0.4em\relax Springer, 2016, pp. 630--645.

\bibitem{chen2021resnests}
K.-L. Chen, C.-H. Lee, H.~Garudadri, and B.~D. Rao, ``{ResNEsts} and {DenseNEsts}: Block-based {DNN} models with improved representation guarantees,'' in \emph{Advances in Neural Information Processing Systems (NeurIPS)}, vol.~34, 2021.

\bibitem{smith2019super}
L.~N. Smith and N.~Topin, ``Super-convergence: Very fast training of neural networks using large learning rates,'' in \emph{Artificial Intelligence and Machine Learning for Multi-Domain Operations Applications}, vol. 11006.\hskip 1em plus 0.5em minus 0.4em\relax SPIE, 2019, pp. 369--386.

\bibitem{paszke2019pytorch}
A.~Paszke, S.~Gross, F.~Massa, A.~Lerer, J.~Bradbury, G.~Chanan, T.~Killeen, Z.~Lin, N.~Gimelshein, L.~Antiga, , A.~Desmaison, A.~K\"{o}pf, E.~Yang, Z.~DeVito, M.~Raison, A.~Tejani, S.~Chilamkurthy, B.~Steiner, L.~Fang, J.~Bai, and S.~Chintala, ``Pytorch: An imperative style, high-performance deep learning library,'' in \emph{Advances in Neural Information Processing Systems (NeurIPS)}, vol.~32, 2019.

\end{thebibliography}
